\documentclass[10pt,twocolumn,twoside]{IEEEtran}

\usepackage{amsmath,amssymb,amsthm}
\usepackage{algorithmic}
\usepackage[dvips]{graphicx}
\usepackage{bm}
\usepackage{color}
\usepackage{float}
\newtheorem{theorem}{Theorem}

\newtheorem{lemma}{Lemma}

\newtheorem{proposition}{Proposition}

\newtheorem{corollary}{Corollary}

\newtheorem{definition}{Definition}

\newcommand{\bc}{\mathbf{c}}

\newcommand{\bh}{\mathbf{h}}

\newcommand{\bW}{\mathbf{W}}

\newcommand{\bv}{\mathbf{v}}

\newcommand{\bI}{\mathbf{I}}
\newcommand{\bw}{\mathbf{w}}

\newcommand{\bx}{\mathbf{x}}

\newcommand{\bA}{\mathbf{A}}
\newcommand{\bB}{\mathbf{B}}

\newcommand{\bH}{\mathbf{H}}

\begin{document}

\title{An Efficient Global Algorithm for Single-Group Multicast Beamforming}

\author{ \IEEEauthorblockN{Cheng Lu, Ya-Feng Liu, \emph{Member, IEEE}}
\thanks{{This work is partially supported by the National Natural
Science Foundation of China under Grants 11671419, 11571221, 11331012, and 11631013}. C. Lu is with the School of Economics and Management, North China Electric Power University, Beijing, 102206, China
(e-mail: lucheng1983@163.com). Y.-F. Liu is with the State Key Laboratory
of Scientific and Engineering Computing, Institute of Computational
Mathematics and Scientific/Engineering Computing, Academy of
Mathematics and Systems Science, Chinese Academy of Sciences,
Beijing, 100190, China (e-mail: yafliu@lsec.cc.ac.cn).}
}


\maketitle

\begin{abstract}
Consider the single-group multicast beamforming problem, where multiple users receive the same data stream simultaneously from a single transmitter. The problem is NP-hard and all existing algorithms for the problem either find suboptimal approximate or local stationary solutions. In this paper, we propose an efficient branch-and-bound algorithm for the problem that is guaranteed to find its global solution. To the best of our knowledge, our proposed algorithm is the first tailored global algorithm for the single-group multicast beamforming problem. Simulation results show that our proposed algorithm is computationally efficient (albeit its theoretical worst-case iteration complexity is exponential with respect to the number of receivers) and it significantly outperforms a state-of-the-art general-purpose global optimization solver called Baron. Our proposed algorithm provides an important benchmark for performance evaluation of existing algorithms for the same problem. By using it as the benchmark, we show that two state-of-the-art algorithms, semidefinite relaxation algorithm and successive linear approximation algorithm, work well when the problem dimension (i.e., the number of antennas at the transmitter and the number of receivers)  is small but their performance deteriorates quickly as the problem dimension increases.

\end{abstract}


\IEEEpeerreviewmaketitle

\section{Introduction}

Physical layer multicasting via transmit beamforming has been recognized as a powerful technique for efficient audio and video streaming in multi-user multi-antenna wireless communication networks. For instance, multicast beamforming is a part of the Evolved Multimedia Broadcast Multicast Service (eMBMS) in
the Long-Term Evolution (LTE) standard, {commercially known as LTE Broadcast \cite{LTEB}}. Multicast beamforming exploits channel state information at the transmitter and utilizes multiple transmit antennas to broadcast common information to a preselected group of users.

%
%
%

One scenario of particular interest in this paper is single-group multicast beamforming, where all users receive the same data stream from the transmitter and the data rate is determined by the minimum received signal-to-noise-ratio (SNR). The earliest mathematical formulation of the single-group multicast beamforming problem is to maximize the average SNR subject to the total transmission power constraint \cite{Lopez}.
This formulation is simple, and can be solved efficiently. However, the solution that maximizes the average SNR does not consider the SNR
of each individual user, so that the minimum received SNR may be significantly lower than the average SNR. Hence, the solution obtained by maximizing the average SNR
does not always achieve a satisfactory common data rate. To overcome this drawback, reference \cite{Sidiropoulos1} proposed two new problem formulations, the quality of service (QoS) constrained problem formulation and the max-min fairness problem formulation, where the former one minimizes the total transmission power subject to SNR constraints of all receivers and the latter one maximizes the minimum SNR among all users subject to the total transmission power constraint. These two new formulations can guarantee QoS of each user and are shown to be equivalent from an optimization point of view \cite{Sidiropoulos1}.
Unfortunately, the two problems are NP-hard in general \cite{Sidiropoulos} and thus it is impossible to solve them to global optimality in polynomial time (unless P=NP) \cite{Garey}. {This is sharply different from the downlink unicast beamforming problem \cite{Schubert,yu2007transmitter,wiesel2006linear} and the uplink/downlink coordinated beamforming problem \cite{dahrouj2010coordinated,liu2013maxsimo,liu11tspbeamforming}, which can be equivalently reformulated as a convex problem and can be efficiently solved to global optimality such as by using the uplink-downlink duality theory.}

%

Various algorithms have been proposed to solve the single-group multicast beamforming problem; see \cite{Sidiropoulos,Tran,Demir,Gopalakrishnan,Lozano,Matskani,Abdelkader,kim2011optimal,konar2016fast,Gopalakrishnan1,Xu,noam2013one,huang2012robust,wen2012rank,schad2012convex,wu2012rank,wu2013physical} and references therein. Based on our knowledge, these algorithms are either convex \emph{relaxation} based algorithms or convex \emph{approximation} based algorithms. Due to the NP-hardness of the problem, none of them {are guaranteed to} find the global solution of the problem (except for very special problem instances \cite{kim2011optimal,Luo2010}). Moreover, since there might exist large gaps between the original problem and its convex relaxations/approximations, the quality of the returned {solutions} by the aforementioned algorithms based on these convex relaxations/approximations might be poor.

The goal of this paper is to develop a computationally efficient global algorithm for solving the single-group multicast beamforming problem that is guaranteed to find the global solution of the problem. 
%
%
%
%

\subsection{Related Works}

One of the state-of-the-art algorithms for the single-group multicast beamforming problem is the semidefinite relaxation (SDR) algorithm \cite{Sidiropoulos}. The main observation behind the SDR algorithm is that the beamforming problem can be equivalently formulated as a rank-one constrained semidefinite program (SDP). The SDR algorithm drops the rank-one constraint, solves the SDR, and then applies a Gaussion randomization strategy to generate a rank-one approximate solution of the original problem based on the obtained solution of the SDR. The SDR algorithm is capable of finding high quality approximate solutions when the number of antennas at the transmitter and the number of
users are small. However, the performance of the SDR algorithm deteriorates quickly as the number of users increases. As justified in \cite{luo2007approximation}, the provable (worst-case) approximation accuracy of the solution returned by the SDR algorithm degrades linearly with the number of users.

One of the best algorithms for the problem of interest is the successive linear approximation (SLA) algorithm \cite{Tran}. At each iteration, the SLA algorithm approximates the original nonconvex problem by using its first-order Talor series expansion at the current iterate and then solves the resulting convex quadratic program to generate the next iterate.
It has been shown in \cite{Tran} that the sequence of points generated by the SLA algorithm converges to a {Karush-Kuhn-Tucker (KKT)} point of the original beamforming problem. Moreover, numerical simulation results in \cite{Tran} have shown that the SLA algorithm performs better than the SDR algorithm {(in the sense that the SLA algorithm finds a better solution with a lower objective value than the SDR algorithm).}

Recently, an alternating maximization (AM) algorithm for the single-group multicast beamforming problem was proposed in \cite{Demir}. Recall that the beamforming problem is equivalent to an SDP with a rank-one constraint \cite{Sidiropoulos}. The main contribution of \cite{Demir} is that it reformulated the problem as a nonconvex problem without the rank-one constraint, which is naturally amenable to AM. Numerical simulation results in \cite{Demir} show that the AM algorithm performs better than the SDR algorithm. However, further simulation results in \cite{Gopalakrishnan} show that the AM algorithm performs worse and has higher complexity than the SLA algorithm.





We also mention some interesting related works, which extend the previously mentioned works \cite{Sidiropoulos,Tran,Demir} on the single-group multicast beamforming problem from different aspects. One extension is to develop more efficient algorithms for solving the problem when computational efficiency is a big issue, i.e., when the channel changes fast over time. Along this direction, low-complexity algorithms \cite{Lozano,Matskani,Abdelkader,kim2011optimal,konar2016fast} as well as adaptive (online) algorithms (which learn the channel correlation matrices) \cite{Gopalakrishnan1,Xu,noam2013one,huang2012robust} have been proposed. These algorithms generally are very fast, but their performance is worse than the ones of the SDR algorithm \cite{Sidiropoulos} and the SLA algorithm \cite{Tran}. Recently, a new adaptive multiplicative update (MU) algorithm and a hybrid MU-SLA algorithm based on it have been proposed in \cite{Gopalakrishnan}. The MU-SLA algorithm runs the SLA algorithm for only one iteration with the returned point by the MU algorithm as the initial point.


Another extension is to develop new physical-layer transmit strategies instead of using the beamforming transmit strategy. Along this direction, a beamformed Alamouti scheme has been independently proposed in \cite{wen2012rank,schad2012convex,wu2012rank}, which can be seen as a rank-two generalization of the previous (rank-one) SDR beamforming framework. It has been shown in \cite{wu2013physical} that the worst-case
approximation accuracy of the beamformed Alamouti scheme
degrades only at a rate of the square root of the number of the users. This improves over the beamforming strategy, where the approximation accuracy degrades at a rate of the number of users \cite{luo2007approximation}.

\subsection{Our Contributions}

This is the first paper that proposes a tailored efficient global algorithm for solving the single-group multicast beamforming problem, which is in sharp contrast to all existing works that focus on the design of approximation algorithms or local optimization algorithms. Our proposed algorithm is based on the branch-and-bound strategy combined with a new argument cut technique. {The argument cuts (in Definition \ref{definition-argcut}) are used to design effective convex relaxations of nonconvex constraints in the single-group multicast beamforming problem and therefore play an important role in our proposed branch-and-bound algorithm for solving the problem.}

Since the single-group multicast beamforming problem is NP-hard, there does not exist a polynomial time algorithm which can solve it to global optimality (unless P=NP) \cite{Garey}. Therefore, our proposed algorithm has an exponential worst-case iteration complexity with respect to the number of users (see Theorem \ref{thm-complexity}). However, our simulation results show that our proposed algorithm is highly efficient and it significantly outperforms the state-of-the-art general-purpose global optimization solver Baron \cite{Tawarmalani,baron2}. The high efficiency of our proposed algorithm is mainly due to the new argument cuts. More specifically, our proposed algorithm can globally solve small-scale problem instances with $2$ antennas at the transmitter and $8$ users within $0.3$ seconds on average while Baron needs $110.4$ seconds; our proposed algorithm can solve median-scale problem instances with $4$ antennas at the transmitter and $8$ users within $2.8$ seconds on average while Baron fails to solve the same problem instances within $10$ minutes. 

Even though in some scenarios (i.e., both the number of antennas at the transmitter and the number of users are large) it takes our proposed algorithm a (relatively) large number of iterations (and thus long time) to find the global solution, our proposed algorithm is still attractive from the following two aspects. First, when the channel changes slowly with time and computational efficiency is not an issue, our proposed algorithm is capable of finding the global solution and therefore provides (potentially much) better performance compared to the existing algorithms. Second, when the channel changes fast with time and computational efficiency becomes an issue in this case, our proposed algorithm can serve as a benchmark to evaluate the performance of the existing heuristic/local optimization algorithms. 
Our simulation results show that the relative gaps between the objective values at the solutions returned by the SDR algorithm and the SLA algorithm and the optimal objective value (achieved by our proposed algorithm) exceed 50\% for some problem instances with a (relatively) large number of antennas at the transmitter and a (relatively) large number of receivers.

\subsection{Organization and Notations}

The organization of this paper is as follows. In Section II, we introduce the QoS constrained formulation of the single-group multicast beamforming problem, and
review two state-of-the-art algorithms for solving the problem.
In Section III, we propose a global branch-and-bound algorithm for solving the problem\footnote{The MATLAB codes of the proposed algorithm are available at [https://www.dropbox.com/s/safrgm97emdgyl9/ARC-BB.rar?dl=0].}. Simulation results are presented in Section IV to illustrate the efficiency of our proposed algorithm and the paper is concluded in Section V.

%

We adopt the following notations in this paper. We use lowercase boldface and uppercase boldface letters to denote (column) vectors and matrices, respectively.
We use  
$\mathbb{C}$ ($\mathbb{R}$) to denote the complex (real) domain, 
$\mathbb{C}^n$  ($\mathbb{R}^n$) to denote the set of
the $n$-dimensional complex (real) column vectors, and $\mathbb{C}^{m\times n}$ to
denote the set of $m\times n$ complex matrices. {For a given complex number $c,$ the notations $\mathrm{Re}(c),~\mathrm{Im}(c),~{|c|},$ and $\arg{(c_i)}$ stand for its real part, its imaginary part, {its modulus}, and its argument, respectively.
For a given (complex) vector $\bx$, $\|\bx\|$ denotes its Euclidean norm, $\bx^H$ denotes its Hermitian, and $\bx^T$ denotes its transpose. The same notations $\bA^H$ and $\bA^T$ also apply to the matrix $\bA.$ For a given complex Hermitian matrix $\bA$, $\bA\succeq \mathbf{0}$ means $\bA$ is positive semidefinite, $\textrm{Trace}(\bA)$ denotes the trace of the matrix $\bA$, and $\textrm{Rank}(\bA)$ denotes the rank of the matrix $\bA$. For two given Hermitian matrices $\bA$ and $\bB$, $\bA\succeq \bB$ means $\bA-\bB\succeq \mathbf{0}$. Finally, we use $\mathbf{i}$ to denote the imaginary unit which satisfies the equation $\mathbf{i}^2 = -1$ and use $\bI_N$ to denote the $N\times N$ identity matrix. 

\section{Problem Formulation and Review}
In this section, we first introduce the QoS constrained formulation of the single-group multicast beamforming problem and then review two state-of-the-art algorithms for solving the problem. 

\subsection{Problem Formulation}
Consider the single-group multicast beamforming problem for the multi-user multi-input single-output (MISO)
downlink channel, where the base station (transmitter) is equipped with $N$ antennas,
and broadcasts a common data stream to $M$ users with a single antenna. Let $\bh_k \in \mathbb{C}^N$
denote the channel vector between the base station and the $k$-th receiver and let $\bw\in \mathbb{C}^N$ denote the beamforming vector used by the base station. Assume that $s(t)$ {is} the broadcasting
data stream. Then the transmitted signal by the base station is given by $s(t) \bw$ and the received signal
at the $k$-th receiver is given by $$y_k(t)=s(t)\bh_k^H\bw +n_k(t),$$ where $n_k(t)$ is the additive white Gaussian noise (AWGN) with
variance $\sigma_k^2$.
Assume that the transmitted signal $s(t)$ has a unit power. Then, the SNR of the $k$-th user can be written as
$$\textrm{SNR}_k=\frac{|\bh_k^H \bw|^2}{\sigma_k^2},~k=1,2,\ldots,M.$$

This paper is interested in minimizing the total transmission power at the base station while satisfying the SNR constraints of all users.
Mathematically, the problem can be formulated as follows:
\begin{equation}\begin{array}{cl}
\displaystyle \min_{\bw} & \|\bw\|^2 \\[10pt]
\mbox{s.t.} & \displaystyle \frac{|\bh_k^H\bw|^2}{\sigma_k^2} \geq \gamma_k,~k=1,2,...,M,
\end{array}\end{equation}
where $\gamma_k$ is the desired transmission SNR target of user $k.$
Let $\tilde{\bh}_k=\bh_k/\sqrt{\gamma_k\sigma_k^2}$, then {$$\frac{|\bh_k^H\bw|^2}{\sigma_k^2} \geq \gamma_k\Longleftrightarrow |\tilde{\bh}_k^H\bw|^2 \geq 1.$$}For ease of notation, we drop the \~  {} and study the following single-group multicast beamforming problem in this paper:
\begin{equation}\begin{array}{cl}
\displaystyle\min_{\bw} & \|\bw\|^2 \tag{P} \\[10pt]
\mbox{s.t.} &{ \displaystyle|\bh_k^H\bw|^2\geq 1},~k=1,2,\ldots,M.
\end{array}\end{equation}

A closely related problem is to maximize the minimum SNR among all users subject to the total transmission power constraint:
\begin{equation}\label{max-min}\begin{array}{cl}
\displaystyle\max_{\bw} & \displaystyle\min_{k=1,2,\ldots,M}\left\{|\bh_k^H\bw|^2\right\} \\[10pt]
\mbox{s.t.} & \displaystyle\|\bw\|^2\leq 1.
\end{array}\end{equation} It has been shown in \cite{Sidiropoulos1} that problems (P) and \eqref{max-min} are equivalent from an optimization perspective of view, i.e., the global solution of problem (P) can be obtained by appropriately scaling the global solution of problem \eqref{max-min}, and vice versa. Therefore, we focus on problem (P) in this paper.


\subsection{Review of Two State-of-the-Art Algorithms}

In this subsection, we briefly review two state-of-the-art algorithms for solving problem (P). We shall compare our proposed algorithm with these two algorithms later in Section IV.

One of the state-of-the-art algorithm for solving problem (P) is the SDR algorithm \cite{Sidiropoulos}. The SDR algorithm is based on the SDR technique, which has been widely used to solve optimization problems arising from signal processing and wireless communications (see \cite{Luo2010} and references therein).
The main observation in the SDR algorithm is that the constraint {$|\bh_k^H\bw|^2\geq 1$} can be equivalently
rewritten as $$\textrm{Trace}(\bH_k \bW)\geq 1,~\bW\succeq\mathbf{0},~\text{and}~\textrm{Rank}(\bW)=1,$$
{where $\bH_k= \bh_k \bh_k^H\in \mathbb{C}^{N\times N}$ and $\bW = \bw \bw^H\in \mathbb{C}^{N\times N}.$} By dropping the rank-one constraint, problem (P) is relaxed to
\begin{equation}\label{SDR}\begin{array}{cl}
\displaystyle\min_{\bW} &\textrm{Trace}(\bW) \\[5pt]
\mbox{s.t.} &\displaystyle\textrm{Trace}(\bH_k \bW)\geq 1,~k=1,\ldots,M,\\[5pt]
&\bW\succeq \mathbf{0}.
\end{array}\end{equation}
The above SDR problem \eqref{SDR} can be solved in polynomial time by using the interior-point algorithm \cite{Ben-Tal}. If the optimal solution
$\bW^*\succeq \mathbf{0}$ of the SDR problem \eqref{SDR} is of rank one, i.e., $\bW^*$ admits the decomposition $\bW^*={\bw^*}(\bw^*)^H,$ then ${\bw^*}$ is a global solution of problem (P).
However, the solution $\bW^*$ of the SDR problem \eqref{SDR} is not always of rank one and thus the global solution of problem (P) might not be obtained. In this case, the SDR algorithm employs the Gaussian randomization techniques to randomly generate approximate solutions based on $\bW^*,$ then scales the approximate solutions to satisfy all SNR constraints, and finally picks the one that has the smallest norm as the final solution. The SDR algorithm can find high quality approximate solutions when the problem dimension, especially the number of users, is small. The worst-case approximation accuracy of the SDR algorithm was shown in \cite{Sidiropoulos,luo2007approximation}.
%

%
%
%
%


Another state-of-the-art algorithm for solving problem (P) is the SLA algorithm \cite{Tran}. The basic idea of the SLA algorithm is to approximate the nonconvex constraint {$|\bh_k^H\bw|^2\geq 1$} by a linear constraint. Specifically, the SLA algorithm introduces the auxiliary variables $$\bv_k:=\left[\mathrm{Re}(\bh_k^H\bw), \mathrm{Im}(\bh_k^H\bw)\right]^T,~k=1,2,\ldots,M.$$ 
Given the current point $\left\{\bv_k^n\in \mathbb{R}^2\right\}$ (at the $n$-th iteration), the SLA algorithm first approximates the nonconvex constraint  {$|\bh_k^H\bw|^2\geq 1$} by the linear constraint
$$\|\bv_k^n\|^2 + 2 (\bv_k^n)^T(\bv_k-\bv_k^n) \geq 1,$$
and then solves the following approximation problem to obtain the next iterate $\left\{\bv_k^{n+1}\right\}:$
\begin{equation}\label{LCQP}\begin{array}{cl}
\displaystyle \min_{\bw,\,\bv}&\|\bw\|^2\\[5pt]
\mbox{s.t.}& \displaystyle\|\bv_k^n\|^2 + 2 (\bv_k^n)^T(\bv_k-\bv_k^n) \geq 1,~k=1,...,M,\\[5pt]
&\displaystyle\bv_k=\left[\mathcal{R}(\bh_k^H\bw), \mathcal{I}(\bh_k^H\bw)\right]^T,~k=1,...,M,
\end{array}\end{equation} where $\bv$ is a collection of $\left\{\bv_k\right\}_{k=1}^M.$ 
The convergence of the SLA algorithm to a KKT point has been established in \cite{Tran}.
It is worthwhile remarking that the subproblem \eqref{LCQP} in the SLA algorithm is a linearly constrained convex quadratic program and thus can be solved efficiently to global optimality \cite{Ben-Tal}. Moreover, for any given $\bv_k^n,$ there holds
$$\|\bv_k\|^2\geq \|\bv_k^n\|^2 + 2 (\bv_k^n)^T(\bv_k-\bv_k^n),~\forall~\bv_k.$$ Therefore, the feasible region of subproblem \eqref{LCQP} is a subset of that of the original problem (P) and the SLA algorithm is a convex (inner) approximation algorithm. This differs from the SDR algorithm which is a convex relaxation algorithm. One potential drawback of the SLA algorithm is that its performance depends on the choice of the initial point $\left\{\bv_k^0\right\}.$ To overcome this, \cite{Tran} proposed to randomly generate many points, scale them, and pick the best one as the initial point.

%
%
%
%

\section{Proposed Global Branch-and-Bound Algorithm}

In this section, we propose a global optimization algorithm for solving problem (P).
Our proposed algorithm is based on the branch-and-bound scheme, which is a general framework
for designing global optimization algorithms \cite{Horst}.

A typical branch-and-bound algorithm (for the minimization problem) is generally based on an enumeration procedure,
which partitions the feasible region to smaller subregions and constructs sub-problems over the partitioned subregions
recursively. In the enumeration procedure, a lower bound
for each subproblem is estimated by solving a relaxation problem. Meanwhile, an upper bound is
obtained from the best known feasible solution generated by the enumeration procedure or by some
other local optimization/heuristic algorithms. A subproblem with a lower bound being larger than the obtained upper bound is
called as an inactive subproblem, which does not contain the global solution of the original problem in its feasible region,
and thus will not be further enumerated. The procedure terminates until all active subproblems have been enumerated, and then an optimal solution within a given relative error
tolerance can be obtained.

The efficiency of a branch-and-bound algorithm considerably relies on the quality of the lower bound as well as the upper bound. The quality of the lower bound depends on the tightness of the convex relaxation and the one of the upper bound depends on the local optimization or heuristic algorithms which are employed to generate the feasible solutions (to the original problem). With better lower and upper bounds, more inactive subproblems can be detected and more unnecessary enumerations
can be avoided.

In the remaining part of this section, we first propose a convex quadratic programming relaxation in Section III-A, which
provides efficient lower bounds in our branch-and-bound algorithm. Then, we present our proposed branch-and-bound algorithm for solving problem (P) in Section III-B. Finally, we show that our proposed branch-and-bound algorithm indeed can find the global solution of problem (P) within any given positive relative error tolerance and analyze its worst-case iteration complexity in Section III-C.


\subsection{New Argument Cut based Relaxation}

As is well known, the SDR \eqref{SDR} is the tightest convex relaxation of problem (P) and its optimal value provides high quality lower bounds on that of problem (P). However, the decision variable in the SDR is lifted to an $N\times N$ matrix, whose dimension
is much larger than the dimension $N$ of the original variable $\bw$. 
In comparison, the subproblem \eqref{LCQP} is a convex linearly constrained quadratic program with $N$ variables, which can be solved much more efficiently than the SDR \eqref{SDR}. However, the subproblem \eqref{LCQP}
is a convex inner approximation of problem (P) but not a convex relaxation, which makes its objective value not suitable for
serving as a lower bound. In the next, we propose a new convex relaxation, which achieves
much higher computational efficiency than the SDR and provides valid lower bounds with satisfactory tightness.

{Without loss of generality, we first change the constraint $|\bh_M^H\bw|^2 \geq 1$ (which is equivalent to $|\bh_M^H\bw| \geq 1$) to $\bh_M^H\bw \geq 1$ in problem (P).} This is because that, for any $\bw$ satisfying $|\bh_M^H\bw|\geq 1,$ we can always find an appropriate $\theta\in \mathbb{R}$ such that $\exp({\textbf{i}\theta})\bh_M^H\bw\geq 1.$ Second, we introduce a new variable $$\bc=\left[c_1,c_2,\ldots,c_{M-1}\right]^T$$ with $c_k=\bh_k^H\bw$,
and represent the constraint $|\bh_k^H\bw|^2\geq 1$ by $|c_k|^2\geq 1$, which is equivalent to $|c_k|\geq 1$.
Then, problem (P) is transformed to the following problem (P'):
\begin{equation}\begin{array}{cl}
\displaystyle \min_{\bw,\,\bc}~&\|\bw\|^2 \tag{P'}\\[5pt]
\mbox{s.t.} & c_k=\bh_k^H\bw,~k=1,...,M-1,\\[5pt]
&|c_k|\geq 1,~k=1,...,M-1,\\[5pt]
&\bh_M^H\bw\geq 1.
\end{array}\end{equation}
In problem (P'), $|c_k|\geq 1~(k=1,...,M-1)$ are the only nonconvex constraints. To develop the branch-and-bound algorithm, we need to relax these nonconvex constraints to convex ones. {It is worthwhile remarking here that we treat the last SNR constraint different from the others in problem (P'). The purpose of doing this is to reduce the number of nonconvex constraints in it and this will make the corresponding relaxation problem more tight and thus improve the efficiency of our proposed branch-and-bound algorithm.}

Let us consider the
set \begin{equation}\label{Cset}\{c_k\in \mathbb{C}\,|\,|c_k|\geq 1\}.\end{equation}
It is simple to see that the convex envelope of set \eqref{Cset} is the whole complex set $\mathbb{C}$. Obviously, it is loose, if {we} directly relax set \eqref{Cset} to its convex envelope, i.e., relax the constraint $|c_k|\geq 1$ to $c_k\in \mathbb{C}$.

Next, we develop a tighter convex relaxation for \eqref{Cset}. To do so, we introduce $x_k=\mathrm{Re}(c_k)$ and $y_k=\mathrm{Im}(c_k)$, and assume that the argument of $c_k$ satisfies $\arg{(c_k)}\in [l_k,u_k]$.  Let

\begin{equation}\label{Dk}\mathcal{D}_{[l_k,u_k]}=\left\{ (x_k,y_k)\,\left|\, \begin{array}{@{}lll}c_k=x_k+y_k\textbf{i},~|c_k| \geq 1\\ \arg{(c_k)}\in [l_k,u_k]\end{array}\right.\right\}\end{equation}

and let
$\textrm{Conv}(\mathcal{D}_{[l_k,u_k]})$ be the convex envelope of the set $\mathcal{D}_{[l_k,u_k]}.$

{Let us first give an illustration on how $\mathcal{D}_{[l_k,u_k]}$ and $\textrm{Conv}(\mathcal{D}_{[l_k,u_k]})$ look like by using Fig. 1 where $[l_k,u_k]=[0,\pi/2].$ In Fig. 1, $$\mathcal{D}_{[0,\pi/2]}=\left\{(x,y)\,|\,x \geq 0,~y\geq 0,~\sqrt{x^2+y^2} \geq 1\right\}$$ is the light blue region outside the arc AB, which is obviously a nonconvex set; and $$\textrm{Conv}(\mathcal{D}_{[0,\pi/2]})= \{ (x,y)\,|\,x \geq 0,~y\geq 0,~x+y \geq 1\}$$ is the light blue region, which is convex and is determined by three lines $\overline{\text{OA}}$, $\overline{\text{OB}}$, and $\overline{\text{AB}}$.}


The following proposition characterizes $\textrm{Conv}(\mathcal{D}_{[l_k,u_k]})$ in the general case if $u_k-l_k\leq \pi.$
\begin{proposition}\label{prop_envelope}
Suppose $l_k$ and $u_k$ in \eqref{Dk} satisfying $u_k-l_k\leq \pi$. Then {
\begin{equation}\label{convD}
\textrm{Conv}(\mathcal{D}_{[l_k,u_k]})= \left\{ (x,y)\,\left|\, \begin{array}{@{}lll}\sin(l_k)x-\cos(l_k)y \leq 0\\ \sin(u_k)x-\cos(u_k)y \geq 0\\ a_kx+b_ky\geq a_k^2+b_k^2\end{array}\right.\right\}.
\end{equation}
}
where
\begin{equation}\label{akbk}a_k=\frac{\cos(l_k)+\cos(u_k) }{2}~\text{and}~b_k=\frac{\sin(l_k)+\sin(u_k) }{2}.\end{equation}
\end{proposition}
\begin{proof}
{It is simple to see that the two extreme points of the set $\textrm{Conv}(\mathcal{D}_{[l_k,u_k]})$ are $(\cos(l_k),\sin(l_k))$ and $(\cos(u_k),\sin(u_k))$. With a little abuse of notation, we denote these two points by B and A (which reduce to the points B and A in Fig. 1 where $[l_k,u_k]=[0,\pi/2]$). To prove the proposition, it is sufficient to determine the three lines $\overline{\text{OA}}$, $\overline{\text{OB}}$, and $\overline{\text{AB}}.$ The equations of these three lines are given by $\sin (l_k) x- \cos (l_k) y = 0,$ $\sin (u_k) x- \cos (u_k) y = 0,$ and $a_k x+b_k y = a_k^2+b^2_k,$ respectively, where $a_k$ and $b_k$ are given in \eqref{akbk}. The proof is completed.}
%
%
\end{proof}



The linear inequalities in \eqref{convD} are introduced due to the
argument constraints. Hence, we name these linear inequalities the argument cuts in this paper.
\begin{definition}[Argument Cuts]\label{definition-argcut}
  The linear inequalities in $\textrm{Conv}(\mathcal{D}_{[l_k,u_k]})$ are called the argument cuts.
\end{definition}
To the best of our knowledge, the argument cuts have not been used in the literatures. With the help of the argument cuts, we are able to develop efficient convex relaxations for problem (P'). More specifically, it follows from Proposition \ref{prop_envelope} that
$$\mathcal{F}_{[l,u]}:=\left\{x+y\textbf{i}\,|\,(x,y)\in \textrm{Conv}(\mathcal{D}_{[l,u]})\right\}$$
is the convex envelope of the nonconvex set
$$\left\{x+y\textbf{i}\,|\,(x,y)\in \mathcal{D}_{[l,u]}\right\}.$$
Assume $\arg (c_k) \in [l_k,u_k]$ for all $k=1,2,\ldots,M-1$ in problem (P'). Then we obtain the following convex argument cut based relaxation (ACR) of problem (P'):
\begin{equation}\label{ACR}\begin{array}{cl}
\displaystyle\min_{\bw,\,\bc} & \|\bw\|^2\\[5pt]
\mbox{s.t.}& c_k=\bh_k^H\bw,~k=1,\dots,M-1,\\[5pt]
& c_k\in \mathcal{F}_{[l_k,u_k]},~k=1,\ldots,M-1,\\[5pt]
& \bh_M^H\bw \geq 1.
\end{array}\end{equation}

%

%
%


\begin{figure}
\centering
\includegraphics[width=7.0cm]{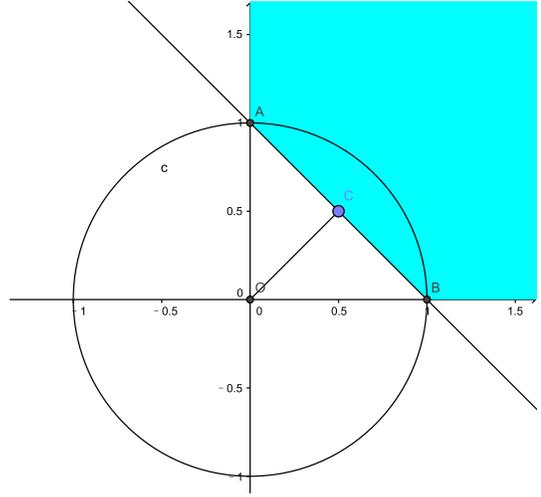}
\centering
\caption{An illustration of the convex envelope $\textrm{Conv}(\mathcal{D}_{[0,\pi/2]})$.}\label{fig-envelope}
\end{figure}




We consider the following two cases: whether or not the width of the interval $[l_k,u_k],$ i.e., $d_k:=u_k-l_k,$ is greater than $\pi.$
\begin{itemize}
  \item [-] Case I: $d_k>\pi.$ In this case, $c_k\in\mathcal{F}_{[l_k,u_k]}$ is equivalent to $c_k\in\mathbb{C}$ and thus the corresponding argument cuts do not take effect. Therefore, we can drop the constraint $c_k\in \mathcal{F}_{[l_k,u_k]}$ in problem \eqref{ACR}.
  \item [-] Case II: $d_k\leq\pi.$ It follows from Proposition \ref{prop_envelope} that $c_k\in \mathcal{F}_{[l_k,u_k]}$
can be represented by at most\footnote{Notice that when $d_k=\pi$, the first two linear constraints in \eqref{convD} are equivalent to each other, and the third linear constraint {becomes} redundant.} three (real) linear inequality constraints. It is simple to see that the argument cuts are effective in this case.
\end{itemize}

To sum up, we know that ACR problem \eqref{ACR} is a linearly constrained convex quadratic program with $2N+2M-2$ (real) variables, at most $2M-1$ (real) linear equality constraints\footnote{Recall that the constraint $\bh_M^H\bw \geq 1$ corresponds a (real) linear inequality constraint and a (real) linear equality constraint.}, and at most $3M-2$ (real) linear inequality constraints. Therefore, ACR problem \eqref{ACR} can be solved efficiently and globally by using the interior-point algorithm within ${\cal{O}}\left(N^3M^{3.5}\right)$ operations \cite[Page 423]{Ben-Tal}. 

%


We conclude this subsection with a tightness measure of the ACR. Let us focus on the case $d_k\leq \pi.$
Note that the inclusion
$$\{x+y\textbf{i}\,|\,(x,y)\in \mathcal{D}_{[l_k,u_k]}\}\supseteq \mathcal{F}_{[l_k,u_k]}$$ generally is not true, which implies that there exist some point $c_k\in\mathcal{F}_{[l_k,u_k]}$ such that $|c_k|<1.$ Therefore, the smallest norm of the points in $\mathcal{F}_{[l_k,u_k]}$ can be used to measure the tightness of the ACR. The following theorem shows that the smallest norm of the points in $\mathcal{F}_{[l_k,u_k]}$ can be computed in a closed form.

\begin{proposition}\label{thm_tightness}
Given any interval $[l_k,u_k]$ with $u_k-l_k\leq \pi$. We have $$\min_{ c_k\in \mathcal{F}_{[l_k,u_k]} } |c_k|= \cos\left(\frac{u_k-l_k}{2}\right).$$
\end{proposition}
\begin{proof}The point that has the smallest norm in a convex set is the projection of the origin onto the corresponding set. We first consider the special case where $[l_k, u_k]=[0,\pi/2],$ as shown in Fig. \ref{fig-envelope}. In this case, the projection of the origin $O$ onto the set $\mathcal{F}_{[0,\pi/2]}$ is the point $C,$ which satisfies $\overline{\text{OC}}\bot \overline{\text{AB}}.$ Notice that the coordinates of $A$ and $B$ in Fig. \ref{fig-envelope} are $(0,1)$ and (1,0), respectively. It is simple to see that the coordinate of the point $C$ is $(0.5, 0.5).$ In the general case of $\mathcal{F}_{[l_k,u_k]},$ the coordinate of the point $c_k$ is $$\left(\frac{\cos(l_k) +\cos(u_k)}{2}, \frac{\sin(l_k) +\sin(u_k)}{2}\right)$$ and its norm is $\cos\left(\frac{u_k-l_k}{2}\right).$ The proof is completed.
%
%
%
\end{proof}

We can see from Proposition \ref{thm_tightness} that: the smaller the width of the interval, the tighter the ACR; as the width of the interval goes to zero, the set $\mathcal{F}_{[l_k,u_k]}$ becomes $\{x+y\textbf{i}\,|\,(x,y)\in \mathcal{D}_{[l_k,u_k]}\}$ and the ACR becomes tight. Hence, an effective approach to tightening the ACR is to reduce the width of the corresponding interval.





\subsection{Proposed Branch-and-Bound Algorithm}
In this subsection, we propose a branch-and-bound algorithm for globally solving problem (P'). The basic idea of the proposed algorithm is to relax the original problem (with appropriate argument constraints) to ACR \eqref{ACR} and gradually tighten the relaxation by reducing the width of the associated intervals.

For ease of presentation, we introduce the following notations. Let $\mathcal{A}=\prod_{k=1}^{M-1} [l_k,u_k]$ and let ACR($\mathcal{A}$) denote the ACR problem defined over the set $\mathcal{A};$ let $\mathcal{P}$ denote the constructed problem list and let $\{\mathcal{A},\bc,L \}$ denote a problem instance from the list $\mathcal{P},$ where $L$ is the optimal value of ACR($\mathcal{A}$) and $\bc$ is its optimal solution; let superscript $t$ denote the iteration number; {let ${U^t}$ denote the upper bound {at the $t$-th iteration}}; and let $\bw^*$ denote the best known feasible solution and let $U^*$ denote the objective value of problem (P') at $\bw^*$.

We are now ready to present the main steps of the proposed branch-and-bound algorithm. 

\textbf{Initialization.} We initialize all intervals $[l^0_k,u^0_k]$ for all $k=1,2,\ldots,M-1$ to be $[0,2\pi],$ i.e., set $\mathcal{A}^0=[0,2\pi]^{M-1}.$ In this case, problem \eqref{ACR} reduces to
\begin{equation}\label{ACR2}\begin{array}{cl}
\displaystyle\min_{\bw} & \|\bw\|^2\\[5pt]
\mbox{s.t.} & \bh_M^H\bw \geq 1.
\end{array}\end{equation} Its optimal solution $(\bw^0, \bc^0)$ and its optimal value $L^0$ are
$$\bw^0=\frac{\bh_M}{\|\bh_M\|^2},~L^0=\frac{1}{\|\bh_M\|^2},$$
$$c_k^0=\frac{\bh_k^H\bh_M}{\|\bh_M\|^2},~k=1,2,\ldots,M-1.$$

\textbf{Termination.} Let $\{\mathcal{A}^t,\bc^t,L^t \}$ be the problem instance that has the least lower bound in the problem list $\mathcal{P}.$ If 
\begin{equation}\label{rgap}({U^{t}}-L^{t})/L^{t}\leq \epsilon,\end{equation} where $\epsilon$ is the preselected relative error tolerance, we terminate the algorithm; otherwise we branch some interval of the above problem instance according to some rule. We can see from \eqref{rgap} that, 
both lower and upper bounds are important to avoid unnecessary branches and enumerations and good lower and upper bounds can significantly improve the efficiency of our proposed algorithm. Below, we shall introduce our branch rule as well as lower and upper bounds one by one.


\textbf{Branch.} Again, let $\{\mathcal{A}^t,\bc^t,L^t \}$ be the problem instance that has the least lower bound in the problem list $\mathcal{P}$ and suppose that the stopping criterion \eqref{rgap} is not satisfied. In this case, we first select the interval that leads to the largest gap to be branched to smaller sub-intervals, i.e., the interval with index
$$k^* = \displaystyle\arg\displaystyle\min_{k\in \{1,...,M-1\}}\left\{|{c}_k^t|\right\};$$ then we partition $\mathcal{A}^t$ into two sets ({denoted as $\mathcal{A}_l^t$ and $\mathcal{A}_r^t$}), with its $k^*$-th interval partitioned into two equal intervals and all the others unchanged. It follows from Proposition \ref{thm_tightness} that the ACR problems defined over the newly obtained two sets, i.e., the two children problems, are tighter than the one defined over the original set $\mathcal{A}^t.$ {Once $\mathcal{A}^t$ has been branched into two sets, the problem instance defined over it will be deleted from the problem list $\mathcal{P}$ and the two children problems will be added into $\mathcal{P}$ if their optimal objective values are less than or equal to the current upper bound.}


\textbf{Lower Bound.} Obviously, for any problem instance $\{\mathcal{A},\bc,L \},$ $L$ is a lower bound of the optimal value of the original problem (P') defined over $\mathcal{A}.$ 
Therefore, the smallest lower bound among all bounds is a lower bound of the optimal value of the original problem. This statement will be formally summarized as Lemma \ref{lemma-bound} and proved in Section III-C.

\textbf{Upper Bound.} An upper bound of the original problem (P') can be obtained by appropriately scaling the solution of any ACR problem instance. More specifically, let $\{\bw^t,\bc^t\}$ be the solution of problem ACR($\mathcal{A}^t$). Then,
\begin{equation}\label{Soperator}\hat\bw^t=\mathrm{Scale}(\bw^t,\bc^{t}):=\frac{\bw^t}{\min\left\{|c_1^t|,|c_2^t|,\ldots,|c_{M-1}^t|,1\right\}}\end{equation}
is feasible to the original problem and $\|\hat \bw^t\|^2$ is an upper bound of the original problem. In our proposed algorithm, the upper bound ${U^t}$ is the best objective values at all of the known feasible solutions {at the $t$-th iteration}.

\begin{figure}[H]
\small{
\textbf{ACR-BB Algorithm for Single-Group Multicast Beamforming Problem (P)}
\begin{algorithmic}[1]
\STATE \textbf{input:} An instance of problem (P), and a relative error tolerance $\epsilon>0.$
\STATE Initialize $\mathcal{P}=\emptyset,$ $\mathcal{A}^0=\prod_{k=1}^{M-1}[l^0_k,u^0_k]=[0,2\pi]^{M-1},$ and set $t=0.$ // \texttt{Initialization.}
\STATE Solve ACR($\mathcal{A}^0$) for its optimal solution $(\bw^0,\bc^0)$ and its optimal value $L^0.$
\STATE Compute $\hat{\bw}^0=\mathrm{Scale}(\bw^0,\bc^0),$ where the operator $\mathrm{Scale}(\cdot,\cdot)$ is defined in \eqref{Soperator}.
\STATE Set ${U^{0}}=\|\hat{\bw}^0\|^2$ and $\bw^{*}=\hat{\bw}^0.$ // \texttt{Initial Upper Bound and Optimal Solution.}
\STATE Add $\{ \mathcal{A}^0,\bc^0,L^0 \}$ into the problem list $\mathcal{P}$.
\LOOP
\STATE Set $t\leftarrow t+1$ and ${U^t=U^{t-1}}.$
\STATE Choose a problem from $\mathcal{P},$ denoted as $\{\mathcal{A}^t,\bc^t,L^t \},$ such that the bound $L^t$ is the smallest one in $\mathcal{P}.$ \label{line:problem} // \texttt{Lower Bound.}
\STATE Delete the chosen subproblem from $\mathcal{P}.$
\IF {$(U^{t}-L^{t})/L^{t}<\epsilon$}
\STATE return ${U^{*}=U^t}$ and $\bw^{*}$ and terminate the algorithm.\label{line:terminate} // \texttt{Termination.}
\ENDIF
\STATE Set $k^* = \displaystyle\arg\displaystyle\min_{k\in \{1,...,M-1\}}\left\{|{c}_k^t|\right\}$ 
and $z^t_{k^*}=\frac{1}{2}(l_{k^*}^t+u_{k^*}^t)$.\label{line:k*}
\STATE Branch $\mathcal{A}^t$ into two sets $\mathcal{A}_{l}^t=\{ \bm{\theta} \in \mathcal{A}^t\,|\,\theta_{k^*}\leq z^t_{k^*} \}$ and
$\mathcal{A}_{r}^t=\{ \bm{\theta} \in \mathcal{A}^t\,|\,\theta_{k^*}\geq z^t_{k^*} \},$ where $\theta_{k^*}$ is the $k^*$-th component of $\bm{\theta}\in\mathbb{R}^{M-1}.$ // \texttt{Branch.}
\STATE Solve ACR($\mathcal{A}_l^t$) for its optimal solution $(\bw_{l}^t,\bc_{l}^t)$ and its optimal value $L_{l}^t$.
\STATE Compute $\hat{\bw}_l^t=\mathrm{Scale}(\bw_l^t,\bc_{l}^t),$ where $\mathrm{Scale}(\cdot,\cdot)$ is defined in \eqref{Soperator}.
\IF {$L_{l}^t\leq {U^{t}}$}
\STATE add $\{ \mathcal{A}_l^t,\bc_l^t,L_l^t \}$ into $\mathcal{P}$.\label{add1}
\ENDIF
\IF{${U^{t}}>\|\hat{\bw}_{l}^t\|^2$}
\STATE set ${U^{t}}=\|\hat{\bw}_{l}^t\|^2$ and $\bw^{*}=\hat{\bw}_{l}^t$. // \texttt{Update Upper Bound and Optimal Solution.}
\ENDIF
\STATE Solve ACR($\mathcal{A}_r^t$) for its optimal solution $(\bw_{r}^t,\bc_{r}^t)$ and its optimal value $L_{r}^t$.
\STATE Compute $\hat{\bw}_r^t=\mathrm{Scale}(\bw_r^t,\bc_{r}^t),$ where $\mathrm{Scale}(\cdot,\cdot)$ is defined in \eqref{Soperator}.
\IF {$L_{r}^t\leq {U^{t}}$}
\STATE add $\{ \mathcal{A}_r^t,\bc_r^t,L_r^t \}$ into $\mathcal{P}$.\label{add2}
\ENDIF
\IF{${U^{t}}>\|\hat{\bw}_{r}^t\|^2$}
\STATE set ${U^{t}}=\|\hat{\bw}_{r}^t\|^2$ and $\bw^{*}=\hat{\bw}_{r}^t$. // \texttt{Update Upper Bound and Optimal Solution.}
\ENDIF
\ENDLOOP
\end{algorithmic}}
\label{alg_wp}
\end{figure}

By judiciously combining the above main steps, we can obtain our proposed branch-and-bound algorithm for solving problem (P) (equivalent to problem (P')). The pseudo-code of our proposed algorithm can be found below. We will call the algorithm ACR-BB for short from now on. 
To make the ACR-BB algorithm more clear, an illustration on how it works is given in Appendix \ref{app-example}. 

We emphasize again that both lower and upper bounds play important roles in the efficiency of the proposed ACR-BB algorithm, because good lower and upper bounds can effectively detect inactive subproblems and avoid unnecessary branches and enumerations. Lines 9 and 11 of the proposed ACR-BB algorithm ensure that it will never branch and enumerate subproblems with the lower bound $L$ being greater than or equal to the upper bound ${U^t}.$ The efficiency of the proposed ACR-BB algorithm will be shown in Section IV.

\subsection{Global Convergence and Worst-Case Iteration Complexity}
In this subsection, we present some theoretical results of our proposed ACR-BB algorithm. We first define the $\epsilon$-optimal solution of problem (P).

\begin{definition}[$\epsilon$-Optimal Solution]
  Given any $\epsilon>0,$ a feasible point $\bw$ is called an $\epsilon$-optimal solution of problem (P) if it satisfies
\begin{equation}
  \frac{\|\bw\|^2-\nu^*}{\nu^*} \leq \epsilon,
\end{equation}where $\nu^*$ is the optimal value of problem (P).
\end{definition}


The following Lemma \ref{lemma-bound} shows that the sequence $\left\{L^t\right\}$ generated by the ACR-BB algorithm is a lower bound of the optimal value of the original problem.

\begin{lemma}\label{lemma-bound}
  For any given instance of problem (P), let $\nu^*$ be its optimal value. Then, we have
  \begin{equation*}\label{Ltnu}0 < L^t \leq \nu^*,~\forall~t\geq 1.\end{equation*}
\end{lemma}
\begin{proof}
The optimal value of problem \eqref{ACR} is always positive because zero is not a feasible solution to it. Hence, $L^t>0$ for all $t\geq 1.$
Next, we show $L^t \leq \nu^*$ for all $t\geq 1.$ 
At the beginning of the $t$-th iteration of the ACR-BB algorithm, the set $\mathcal{A}^0$ has been partitioned into $t$ small subsets, and {the global solution of the given problem instance} must lie in one of them (we denote the corresponding subset as $\mathcal{A}^*$). Then, the optimal value of problem ACR($\mathcal{A}^*$) must be less than or equal to $\nu^*.$ Moreover, since $L^t$ is the smallest lower bound of all subproblems in $\mathcal{P}$ at the $t$-th iteration, it follows that $L^t$ is less than or equal to the optimal value of problem ACR($\mathcal{A}^*$). Therefore, we get $L^t\leq \nu^*$ for all $t\geq 1.$
%
\end{proof}

From lemma \ref{lemma-bound}, we immediately get
$$\frac{{U^t}-L^t}{L^t}\geq \frac{{U^t}-\nu^*}{\nu^*},~\forall~t\geq 1.$$
This further shows that, if the ACR-BB algorithm terminates, i.e., condition \eqref{rgap} is satisfied, the returned solution $\bw^*$ by the algorithm is an $\epsilon$-optimal solution of problem (P).

The following Lemma \ref{lemma-terminate} shows that the ACR-BB algorithm will terminate.

\begin{lemma}\label{lemma-terminate}
For any given instance of problem (P), let $\{ \mathcal{A}^t,\bc^t,L^t \}$ be the subproblem chosen in Line \ref{line:problem} and let $k^*$ be the index chosen in Line \ref{line:k*}. If \begin{equation}\label{uldifference}u^t_{k^*}-l^t_{k^*} \leq 2\delta,\end{equation}
where
\begin{equation}\label{delta}\delta=\arccos \left(\frac{1}{\sqrt{1+\epsilon}}\right),\end{equation}
then condition \eqref{rgap} holds true and the ACR-BB algorithm will terminate in Line \ref{line:terminate}. 
%
\end{lemma}
\begin{proof}
It follows from \eqref{uldifference} and Proposition \ref{thm_tightness} that
$$|c^t_{k^*}| \geq \cos\left(\frac{u^t_{k^*}-l^t_{k^*}}{2}\right) \geq \frac{1}{\sqrt{1+\epsilon}}.$$
By the above inequality and the definition of $k^*$ (see Line \ref{line:k*} of the ACR-BB algorithm), we obtain
\begin{equation}\label{clower}
  |c^t_{i}|\geq |c^t_{k^*}|\geq  \frac{1}{\sqrt{1+\epsilon}},~\forall~i=1,2,\ldots,M-1.
\end{equation}
Let $(\bw^t,\bc^t)$ be the solution of problem ACR($\mathcal{A}^t$). Then, it follows from \eqref{Soperator} and \eqref{clower} that the scaled feasible solution $\hat{\bw}^t=\mathrm{Scale}(\bw^t,\bc^{t})$ satisfies
$$\|\hat{\bw}^t\|^2= \frac{\|\bw^t\|^2}{\min\left\{|c_1^t|^2,|c_2^t|^2,\ldots,|c_{M-1}^t|^2,1\right\}} \leq \|\bw^t\|^2(1+\epsilon).$$ Moreover, since ${U^t}$ is the objective value at the best known feasible solution at the $t$-th iteration, we get \begin{equation}\label{U*upper}{U^t}\leq \|\hat{\bw}^t\|^2 \leq \|\bw^t\|^2(1+\epsilon).\end{equation}
Now, we can use the fact $L^t=\|\bw^t\|^2$ and \eqref{U*upper} to obtain
$$\frac{{U^t}-L^t}{L^t}=\frac{{U^t}-\|\bw^t\|^2}{\|\bw^t\|^2}\leq \frac{\|\bw^t\|^2(1+\epsilon)-\|\bw^t\|^2}{\|\bw^t\|^2}=\epsilon.$$
The proof is completed.
%
%
%
%
\end{proof}

Based on Lemmas \ref{lemma-bound} and \ref{lemma-terminate}, we obtain the main result of this subsection.

\begin{theorem}[Iteration Complexity]\label{thm-complexity}
For any given $\epsilon>0$ and any given instance of problem (P) with $M$ users, the ACR-BB algorithm will return an $\epsilon$-optimal solution of the given instance within at most
\begin{equation}\label{bound}T:=\left\lceil \left(\frac{2\pi}{\delta}\right)^{M-1} \right\rceil+1\end{equation}
iterations, where $\delta$ is defined in \eqref{delta}.


%
\end{theorem}
\begin{proof}
It follows from Lemmas \ref{lemma-bound} and \ref{lemma-terminate} that, if \eqref{uldifference} is satisfied, the ACR algorithm will terminate and return an $\epsilon$-optimal solution of the given problem instance. To show the theorem, it remains to show that the algorithm will terminate within $T$ iterations, where $T$ is defined in \eqref{bound}. Next, we show this based on the contradiction principle.

Suppose that the algorithm does not terminate within $T$ iterations. This fact, together with Lemma \ref{lemma-terminate}, implies that the interval that is chosen to be partitioned at the $t$-th iteration must satisfy
$u^t_{k^*}-l^t_{k^*} > 2\delta$ for all $t=1,2,\ldots,T.$ 
Then, after the partition, the width of the two sub-intervals $[l_{k^*}^t, z_{k^*}^t]$ and $[z_{k^*}^t, u_{k^*}^t]$ is greater than $\delta.$ 
Based on this, we can conclude that, for each subset $\mathcal{A}=\prod_{k=1}^{M-1}[l_k,u_k]$
partitioned from the original set $\mathcal{A}^0$, there holds $u_k-l_k > \delta$ for all $k=1,...,M-1.$ Hence, the volume of each subset $\mathcal{A}$ is not less than $\delta^{M-1}$ and the total volume of all $T$ subsets is not less than $T\delta^{M-1}.$ Obviously, the volume of $\mathcal{A}^0$ is $(2\pi)^{M-1}.$ By the choice of $T,$ we get
$T\delta^{M-1}>(2\pi)^{M-1},$ which further implies that the total volume of all $T$ subsets is greater than the one of the original set $\mathcal{A}^0.$ This is a contradiction. Hence, the algorithm will terminate within at most $T$ iterations.
%
%
\end{proof}
{Some remarks on Theorem \ref{thm-complexity} are in order. First, from Theorem \ref{thm-complexity}, we can immediately obtain the following convergence result, which shows that both the sequences of the upper bounds and the lower bounds generated by the ARC-BB algorithm with $\epsilon=0$ converge to the optimal objective value of problem (P). Therefore, the ACR-BB algorithm with $\epsilon=0$ indeed is a global algorithm and is able to find the global solution of problem (P).
\begin{corollary}[Global Optimality]
  For any given instance of problem (P), let $\nu^*$ be its optimal value and let $\left\{U^t\right\}$ and $\left\{L^t\right\}$ be the iterates generated by the ARC-BB algorithm with $\epsilon =0$. Then, we have $U^t\rightarrow \nu^*$ and $L^t\rightarrow \nu^*.$
\end{corollary}
\begin{proof}
  Theorem \ref{thm-complexity} shows that, for any given $\epsilon>0,$ there exists an integer $T\geq 1$ (given in \eqref{bound}) such that $(U^t-L^t)/L^t\leq \epsilon$ for all $t\geq T.$ This statement is equivalent to ${(U^t-L^t)}/{L^t}\rightarrow 0.$ Since the sequence $\left\{L^t\right\}$ is uniformly bounded away from zero for any given instance of problem (P), we further obtain ${U^t-L^t}\rightarrow 0.$ This, together with Lemma \ref{lemma-bound} and the definition of $U^t$, immediately implies the desired results. The proof is completed. \end{proof}

In practice, we need to preselect a positive relative error tolerance $\epsilon$ in our proposed ACR-BB algorithm, as most of iterative optimization algorithms \cite[Chapter 9]{boyd2004convex}.} Theorem \ref{thm-complexity} shows that the total number of iterations for our proposed ACR-BB algorithm to return an $\epsilon$-optimal solution of any instance of problem (P) is exponential with respect to the number of receivers $M.$
%
The iteration complexity of our proposed algorithm seems high at first sight. However, as will be shown in Section IV, its practical iteration complexity is actually significantly less than the worst-case bound in \eqref{bound}. It is also worthwhile remarking that there is no polynomial time algorithm which can globally solve the problem (unless P=NP) \cite{Garey}, because the problem is NP-hard. 

%


\section{Simulation Results}
In this section, we present some numerical simulation results of our proposed ACR-BB algorithm for solving problem (P). More specifically, we first present some simulation results to show the convergence behaviors of the proposed algorithm in Section IV-A. Then, we show the efficiency of the proposed algorithm by comparing it with a state-of-the-art general-purpose global optimization solver called Baron \cite{Tawarmalani,baron2}. Last but not least, we use the proposed algorithm as the benchmark to evaluate the performance of two state-of-the-art algorithms \cite{Sidiropoulos,Tran} for solving the same problem. 
%

In all of our simulations, 
the channel vectors $\bh_k$ are randomly generated according to the distribution $\mathcal{C}\mathcal{N}(\mathbf{0},\bI_N)$ as done in \cite{Sidiropoulos,Tran,wu2013physical}. 
 We performed all numerical experiments on a PC with a 3.40-GHz Intel Core i7-2600 processor with access to 4 GB of RAM. We implemented our proposed ACR-BB algorithm in Matlab 7.10. We use the built-in function ``quadprog'' in Matlab to solve all linearly constrained convex quadratic
{programs}, and use Sedumi \cite{Sedumi} to solve all SDPs.


\subsection{Convergence Behaviors of the ACR-BB Algorithm}

%
%

In this subsection, we generate a problem instance with $(N, M)=(4, 40),$ apply our proposed ACR-BB algorithm to solve it, and study the convergence behaviors of our proposed algorithm, i.e., the convergence behaviors of the lower bounds $\left\{L^t\right\}$ and {the upper bounds $\left\{U^t\right\}.$} 
In this simulation, we set $\epsilon=1$e$-7$ in our proposed algorithm such that the algorithm can find the ``real'' optimal solution when it terminates. We denote the objective value at the returned solution by $\bar \nu$ and define
$$E_1^t=\frac{|L^t-\bar \nu|}{\bar \nu},~E_2^t=\frac{|U^t-\bar \nu|}{\bar \nu},~\text{and}~E_3^t=\frac{|U^t-L^t|}{L^t}.$$
{In the above, $E_1^t$ and $E_2^t$ are the relative errors of the lower bound $L^t$ and the upper bound $U^t$, respectively;
$E_3^t$ is an easily computable upper bound of $E_1^t+E_2^t$. In this paper, we will call all of them relative errors.} The above three relative errors versus the number of iterations are illustrated as Fig. \ref{fig-gap}.
%

\begin{figure}
\centering
\includegraphics[width=8.0cm]{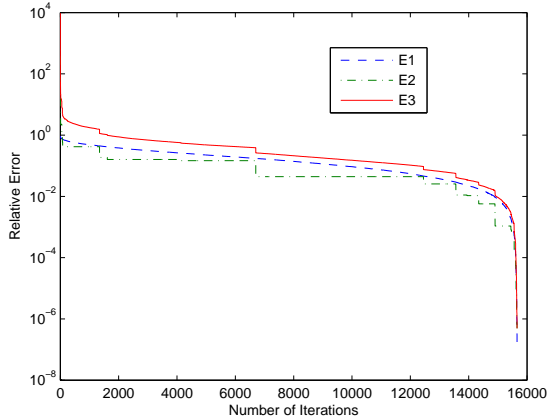}
\centering
\caption{Relative errors versus the number of iterations with $(N, M)=(4, 40)$.}\label{fig-gap}
\end{figure}

It can be seen from Fig. \ref{fig-gap} that $E_2$ converges much faster than $E_1$ and $E_3$ at the beginning stage of our proposed algorithm. For instance, $E_2$ becomes smaller than $1$e$-1$ at the $6702$-th iteration, while $E_1$ and $E_3$ become smaller than $1$e$-1$ at the $9653$-th and $12221$-th iteration, respectively.
The above results imply that an $0.1$-optimal solution has been found at the $6702$-th iteration but it is verified (to be an $0.1$-optimal solution) at the $12221$-th iteration. These results show that:
\begin{itemize}
  \item [-] The upper bounds $\left\{U^t\right\}$ converge faster than the lower bounds $\left\{L^t\right\}$ at the beginning stage of our proposed algorithm; and more importantly,
  \item [-] A low-accuracy solution can be more easier found than verified by our proposed algorithm.
\end{itemize}

We can also see from Fig. \ref{fig-gap} that it takes our proposed ACR-BB algorithm $3339$ iterations to reduce $E_3$ from $1$e$-1$ to $1$e$-3$ and $96$ iterations to reduce it from $1$e$-3$ to $1$e$-6,$ respectively. This shows that the performance of our proposed algorithm in terms of the number of iterations is not sensitive to the choice of the relative error tolerance $\epsilon$ (when it is smaller than $10^{-3}$).

%

In the remaining part of this section, we set $\epsilon=5$e$-3.$


%
%
%
\begin{table}
\centering
\caption{Comparison of our proposed ACR-BB algorithm and Baron} 
\label{table-baron}
\begin{tabular}{|c|rrr|rrr|}
\hline
Instance   & \multicolumn{3}{c|}{ACR-BB}  & \multicolumn{3}{c|}{Baron}  \\
ID  & Value &\# of Iter.  & Time & Value & \# of Iter.  & Time  \\
\hline
01  &1.2344  &90 &0.7 &1.2344 &6297 &327.9 \\
02  &1.1412  &40 &0.2 &1.1412 &2249 &147.4 \\
03  &1.3541  &68 &0.4 &1.3541 &2917 &94.5 \\
04  &1.4251  &49 &0.3 &1.4251 &2847 &261.5 \\
05  &1.4136  &61 &0.4 &1.4136 &1825 &49.8 \\
06  &0.8540  &40 &0.3 &0.8540 &1303 &31.9 \\
07  &0.9777  &48 &0.3 &0.9777 &457 &19.0 \\
08  &5.2469  &1 &0.0 &5.2469 &4305 &58.4 \\
09  &11.9555  &16 &0.1 &11.9555 &5358 &87.1 \\
10  &1.3258  &62 &0.4 &1.3258 &653 &26.5 \\
\hline
Average &2.6928 &47.5 &0.3 &2.6928 &2821.1 &110.4\\
\hline
\end{tabular}
\end{table}

\subsection{Efficiency of the ACR-BB Algorithm}
In this subsection, we study the efficiency of our proposed ACR-BB algorithm for solving problem (P). 

We first compare our proposed algorithm with Baron\footnote{To the best of our knowledge, our proposed ACR-BB algorithm is the first tailored global algorithm for problem (P) and there is no existing global algorithms specially designed for the problem that we can compare our
proposed algorithm with.} \cite{Tawarmalani,baron2}, which is a state-of-the-art general-purpose
global optimization solver and has been widely applied to solve problems arising from various applications \cite{GopalakrishnanA,Konstanteli}. {Although both of our proposed algorithm and Baron lie in the branch-and-bound framework, the difference between them is remarkable. First, the branch strategy in the two algorithms is different, i.e., our proposed algorithm branches in the (argument) ranges of $\arg (c_k)$ but Baron transforms problem (P) into a real (nonconvex) quadratic program by representing the real and imaginary parts of each complex variable with two independent real variables and branches in the ranges of these real variables. 
Second, convex relaxation used in the two algorithms is also different, i.e., our proposed algorithm is based on ACR \eqref{ACR}, which is a linearly constrained quadratic program, but Baron is based on linear programming relaxation \cite{Tawarmalani,baron2}. Since the lower bound based on ACR \eqref{ACR} is generally better than the one based on linear programming relaxation, our proposed algorithm is much more efficient than Baron, as shown below.}
%
%
%

We apply our proposed ACR-BB algorithm and Baron to solve 10 randomly generated problem instances with $(N,M)=(2,8).$ The comparison results are summarized in Table \ref{table-baron}, where the first column shows the {IDs} of the corresponding problem {instances}, the second column shows the objective values and the number of iterations and CPU time results obtained by our proposed algorithm, the third column shows the objective values and the number of iterations and CPU time results obtained by Baron, and the last row shows the average results over the $10$ instances. We can observe from Table \ref{table-baron} that our proposed ACR-BB algorithm performs $366$ times faster than Baron (on average) to find the same solutions. In fact, we have tried to use Baron to solve problem instances with larger $N$ and/or $M.$ Unfortunately, we found that Baron fails to solve all problem instances with $N\geq4$ and $M\geq 8$ within $10$ minutes. These observations demonstrate that our specially designed ACR-BB algorithm achieves significantly higher efficiency on globally solving problem (P) than Baron.


Next, we present more numerical results on applying our proposed ACR-BB algorithm to solve problem (P) with larger $N$ and/or $M$ without comparing it with other general-purpose global optimization solvers. 
All the results to be shown from now on are obtained by averaging over or choosing from $50$ randomly generated problem instances for each pair $(N, M)$. 
Table \ref{table-average-worst} reports the number of iterations and CPU time results, where the first column shows the setup of the problem, the second column shows the average results (over the $50$ problem instances), and the third column shows the worst-case result (among the $50$ problem instances). Figs. \ref{fig-time} and \ref{fig-iteration} plot the average CPU time and the average number of iterations versus the number of users with $N=4.$



\begin{table}
\centering
\caption{Average and worst-case results of number of iterations and CPU time (in seconds) under different setups}\label{table-average-worst}
\begin{tabular}{|r|rr|rr|}
\hline
Setting   & \multicolumn{2}{r|}{Average Performance}  & \multicolumn{2}{r|}{Worst-Case Performance}  \\
$(N,M)$  & \# of Iter.  & Time  & \# of Iter.  & Time  \\
\hline
(2,8)  &47.9     &0.3  &103 &0.7  \\
(2,16) &72.7     &0.5  &147 &1.0  \\
(2,24) &100.6    &0.6  &201 &1.2 \\
(2,32) &136.0    &0.9  &225 &1.4  \\
(2,40) &151.3    &1.0  &298 &2.0 \\
(2,48) &171.7    &1.2  &278 &1.8 \\
(2,56) &199.9    &1.4  &343 &2.3 \\
(2,64) &215.4    &1.5  &341 &2.4 \\
(4,8)  &330.1    &2.8  &886    &8.1   \\
(4,16) &1721.2   &16.6 &3763   &37.8  \\
(4,24) &4492.1   &46.4 &11627  &125.8  \\
(4,32) &8579.6   &98.2  &17442  &216.0 \\
(4,40) &16108.7  &184.9  &34675  &419.9   \\
(6,8)  &1728.6   &21.5   &6388  &83.4  \\
(6,16) &13233.1  &199.9  &36804 &593.3 \\
(6,24) &63691.6  &1159.7 &175041 &3679.3 \\
(8,8)  &2090.0   &30.1   &6301   &95.1  \\
(8,16) &64168.9  &1396.9  &235309 &5786.3 \\
\hline
\end{tabular}
\end{table}

\begin{figure}
\centering
\includegraphics[width=7.4cm]{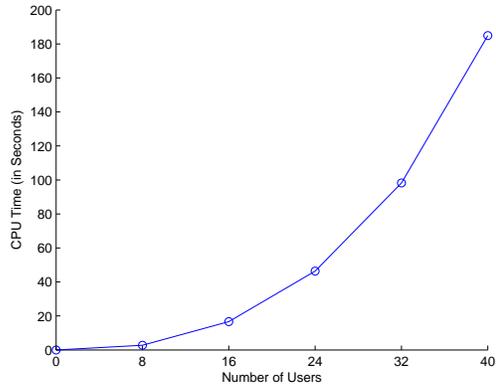}
\centering
\caption{Average CPU time of our proposed ACR-BB algorithm versus the number of users with $N=4$. \label{fig-time}}
\end{figure}

\begin{figure}
\centering
\includegraphics[width=7.4cm]{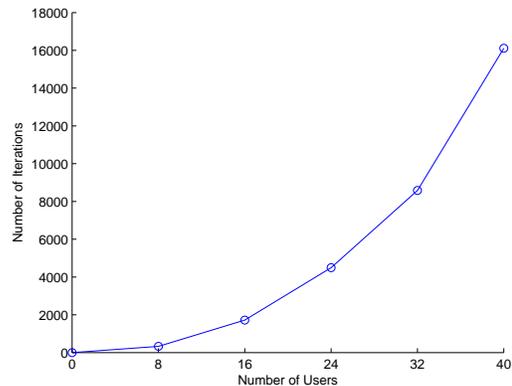}
\centering
\caption{Average number of iterations of our proposed ACR-BB algorithm versus the number of users with $N=4$. \label{fig-iteration}}
\end{figure}

We can observe from the above simulation results (Table \ref{table-average-worst})
that our proposed ACR-BB algorithm is capable of finding the global solution (with the relative error being not greater than $5$e$-3$) within several seconds when $N$ is small and $M$ is not too large. In particular, the proposed algorithm can solve all generated problem instances with $N=2$ and $M$ ranging from $8$ to $64$ within $2.4$ seconds. When $N=4,$ the proposed algorithm can solve problem instances with $M\leq 16$ within $20$ seconds on average and problem instances with $M$ ranging from $24$ to $40$ within $185$ seconds on average. The proposed algorithm can solve all $50$ problem instances with $N=4$ and $M$ ranging from $8$ to $40$ within $7$ minutes. Moreover, we can observe from Fig. \ref{fig-iteration} that the number of iterations of our proposed algorithm is extremely smaller than the worst-case bound in \eqref{bound}. These results demonstrate that our proposed ACR-BB algorithm is highly efficient for problem (P) when $N$ is small and $M$ is not too large. The high efficiency of our proposed ACR-BB algorithm is due to the newly developed ACR \eqref{ACR}, which effectively avoids unnecessary branches that do not contain global optimal solutions.

When $N$ and $M$ (especially $N$) are large, it takes our proposed algorithm relatively long time (and a relatively large number of iterations) to find the global solution. For instance, our proposed algorithm needs $1396.9$ seconds on average to solve problem instances with $(N,M)=(8,16).$ Even though the efficiency of our proposed algorithm is not very high when $N$ and $M$ (especially $N$) are large, our proposed algorithm is still useful in some scenarios. For instance, in the scenario where the channel is stationary/constant and computational efficiency is not a big issue, our proposed algorithm can find the global solution and thus can provide the best performance in terms of transmitting the minimum total power. Another key application of our proposed ACR-BB algorithm is that it provides an important benchmark for performance evaluation of other heuristic/local optimization algorithms for the same problem, as will be done in the next subsection.

\begin{table}
\centering
\caption{Average objective values and relative gaps obtained by SLA and SDR} \label{table-average-two}
\begin{tabular}{|c|rrr|rr|}
\hline
Setting   & \multicolumn{3}{r|}{Average Objective Value}  & \multicolumn{2}{r|}{Average Relative Gap}  \\
$(N,M)$  & ACR-BB  & SLA  & SDR  & SLA  & SDR   \\
\hline
(2,8) &1.629 &1.629 &1.642 &0.02\% &0.79\%  \\
(2,16) &2.800 &2.803 &2.827 &0.10\% &0.96\%  \\
(2,24) &3.389 &3.398 &3.463 &0.26\% &2.20\%  \\
(2,32) &4.457 &4.489 &4.642 &0.73\% &4.16\%  \\
(2,40) &5.720 &5.731 &5.871 &0.20\% &2.65\%  \\
(2,48) &5.679 &5.739 &5.892 &1.06\% &3.77\%  \\
(2,56) &5.461 &5.491 &5.758 &0.55\% &5.44\%  \\
(2,64) &5.914 &5.967 &6.291 &0.89\% &6.37\%  \\
(4,8) &0.514 &0.525 &0.577 &2.01\% &12.17\%  \\
(4,16) &0.837 &0.902 &1.121 &7.83\% &33.97\%  \\
(4,24) &1.132 &1.256 &1.710 &10.95\% &51.03\%  \\
(4,32) &1.328 &1.587 &2.045 &19.48\% &54.04\%  \\
(4,40) &1.525 &1.770 &2.587 &16.08\% &69.64\%  \\
(6,8) &0.334 &0.341 &0.393 &2.30\% &17.75\%  \\
(6,16) &0.531 &0.578 &0.799 &8.90\% &50.39\%  \\
(6,24) &0.667 &0.779 &1.113 &16.77\% &66.99\%  \\
(8,8) &0.246 &0.253 &0.278 &2.95\% &12.93\%  \\
(8,16) &0.376 &0.420 &0.618 &11.68\% &64.11\%  \\
\hline
\end{tabular}
\end{table}

\subsection{Performance Evaluation of Existing Heuristic/Local Optimization Algorithms}
In this subsection, we evaluate the performance of two state-of-the-art algorithms for problem (P), the SDR algorithm \cite{Sidiropoulos} and the SLA algorithm \cite{Tran}, by using our proposed ACR-BB algorithm as the benchmark. {Our evaluation metrics are the objective values and the relative gaps. Note that the objective function in our problem is the total transmission power and hence the smaller the objective value returned by an algorithm, the better the algorithm. The relative gap of an algorithm is defined as $(\nu_A-\bar {\nu} )/\bar{\nu}$, where $\nu_A$ is the objective value returned by the algorithm and $\bar \nu$ is the objective value returned by our proposed algorithm.}

In our simulations, the SDR algorithm and the SLA algorithm are implemented by following the descriptions in \cite{Sidiropoulos} and \cite{Tran}, respectively. In particular, $1000$ points are randomly generated based on the solution of the SDR \eqref{SDR} to obtain the final approximate solution in \cite{Sidiropoulos} and also $1000$ points are randomly generated to obtain a good initial point for the SLA algorithm in \cite{Tran}. Recall that we have generated ${900}$ instances of problem (P) in Table \ref{table-average-worst} and obtained their optimal solution (with the relative errors being not greater than $5$e$-3$) by using our proposed ACR-BB algorithm. We apply the SDR algorithm and the SLA algorithm to solve these instances and evaluate their performance.


%
Table \ref{table-average-two} reports the average objective values and average relative gaps of the SDR algorithm and the SLA algorithm.
We can see from Table \ref{table-average-two} that the average relative gaps of both algorithms are small when $N$ and $M$ (especially $N$) are small and the average relative gaps of the SLA algorithm is smaller than the ones of the SDR algorithm. This shows that, the quality of the returned solutions by both algorithms is high when $N$ and $M$ (especially $N$) are small and the quality of the returned solution by the SLA algorithm is generally better than the one of the returned solution by the SDR algorithm. However, as $N$ and/or $M$ increase, the average relative gaps of both algorithms quickly become large. For instance, when $(N,M)=(4,40)$, the relative gap of the SDR algorithm is $69.64\%$; when $(N,M)=(4,32)$, the relative gap of the SLA algorithm
is $19.48\%$. These results show that the quality of the solutions obtained by the two algorithms are not good when $N$ and/or $M$ are large and the performance of both algorithms degrades quickly as $N$ and/or $M$ increase. 
This makes sense, because both of the algorithms are not global algorithms and they are more likely to get stuck at a local solution when $N$ and/or $M$ are large.


\begin{table}
\centering
\caption{Worst-case relative gaps obtained by SLA and SDR.}\label{table-worst-two}
\begin{tabular}{|c|r|r|}
\hline
$(N,M)$     &SLA &SDR  \\
\hline
(2,8)     &1\% &6\%  \\
(2,16)    &6\% &10\%  \\
(2,24)    &4\% &18\%  \\
(2,32)    &13\% &30\%  \\
(2,40)    &8\% &22\%  \\
(2,48)    &10\% &28\%  \\
(2,56)    &12\% &25\%  \\
(2,64)    &14\% &25\%  \\
(4,8)     &34\% &42\%  \\
(4,16)    &39\% &82\%  \\
(4,24)    &65\% &100\%  \\
(4,32)    &70\% &103\%   \\
(4,40)    &74\% &104\%   \\
(6,8)     &22\% &92\%  \\
(6,16)    &62\% &126\%  \\
(6,24)    &73\% &128\%  \\
(8,8)     &23\% &46\%  \\
(8,16)    &51\% &124\% \\
\hline
\end{tabular}
\end{table}

To study the robustness of the two algorithms, we list the worst-case relative gaps of the two algorithms in Table \ref{table-worst-two}.
We can see from Table \ref{table-worst-two} that the worst-case relative gap of the SLA algorithm exceeds $70\%$ in two setups $(4,40)$ and $(6,24);$ and the worst-case relative gap of the SDR algorithm exceeds $100\%$ in five setups $(4,32),~(4,40),~(6,16),~(6,24),$ and $(8,16).$ These results clearly demonstrate that both of the algorithms are not robust (especially when $N$ and/or $M$ are relatively large) because their performance for some problem instance is bad and the SLA algorithm is more robust than the SDR algorithm.

\begin{table}
\centering
\caption{Probability of obtaining the global solution by SLA and SDR}\label{table-probablity-two}
\begin{tabular}{|c|r|r|}
\hline
$(N,M)$     &SLA &SDR  \\
\hline
(2,8)     &98\% &74\%  \\
(2,16)    &96\% &68\%  \\
(2,24)    &92\% &48\%  \\
(2,32)    &88\% &28\%   \\
(2,40)    &92\% &30\%  \\
(2,48)    &82\% &22\%  \\
(2,56)    &86\% &24\%   \\
(2,64)    &88\% &14\%  \\
(4,8)     &80\% &38\%  \\
(4,16)    &50\% &6\%  \\
(4,24)    &42\% &0\%  \\
(4,32)    &28\% &0\%   \\
(4,40)    &14\% &0\%  \\
(6,8)     &74\% &40\%  \\
(6,16)    &44\% &0\%  \\
(6,24)    &18\% &0\%  \\
(8,8)     &72\% &28\% \\
(8,16)    &42\% &2\% \\
\hline
\end{tabular}
\end{table}

We also test the probability that the two algorithms can find the global solution.
We call that an algorithm finds the global solution of a problem instance if the relative gap is less than or equal to $1$e$-3.$ The probability that an algorithm finds the global solution is defined as the ratio of the number of instances of which it finds the global solution and the total number of instances of which it is applied to solve. The latter is $50$ in our case. The probability that the two algorithms find the global solution is reported in Table \ref{table-probablity-two}. As can be seen from the table, the probability that the SLA algorithm and the SDR algorithm can find the global solution is (relatively) large when both $N$ and $M$ are small and the probability that the SLA algorithm finds the global solution is generally much larger than the one that the SDR algorithm finds the global solution. However, the probability that the SLA algorithm and the SDR algorithm find the global solution decreases quickly as $N$ and/or $M$ increase. These results are intuitive and agree well with our previous simulation results.

In summary, we can make the following conclusions on the performance of the two state-of-the-art algorithms for problem (P), the SLA algorithm and the SDR algorithm, by using our proposed ACR-BB algorithm as benchmark:
\begin{itemize}
  \item [-] both of the algorithms perform well when $N$ and $M$ are small;
  \item [-] the performance of the SLA algorithm is generally better than the one of the SDR algorithm in terms of (average and worst-case) relative gaps, robustness, and probability of finding the global solutions; and
  \item [-] the performance of both of the algorithms degrades quickly as $N$ and/or $M$ increase.
\end{itemize}

{It is worthwhile remarking that both of the SLA algorithm and the SDR algorithm are more efficient than our proposed ACR-BB algorithm. For instance, to solve problem (P) with $N=8$ and $M=16,$ both of the algorithms need less than 0.25 seconds (on average) but it takes our proposed algorithm 1396.9 seconds. However, our proposed ACR-BB algorithm is a global algorithm which is guaranteed to find the global solution of the problem within any given relative error tolerance whereas the SLA algorithm is a local optimization algorithm and the SDR algorithm is an approximation algorithm and both of them cannot guarantee global optimality of their returned solutions.}

\section{Conclusion}
{In this paper, we proposed the ACR-BB algorithm for solving the single-group multicast beamforming problem. The proposed algorithm is guaranteed to find the global solution
of the problem within any given relative error tolerance. To the best of our knowledge, the proposed ACR-BB algorithm is the first specially designed global algorithm for solving the
single-group multicast beamforming problem. The proposed
algorithm is based on the branch-and-bound strategy as well as
a newly developed argument cut technique. Simulation results show that the proposed algorithm significantly outperforms a state-of-the-art general-purpose solver Baron. The high efficiency of the proposed algorithm is mainly due to the use of the argument cut based relaxation, which can potentially be used in other related
algorithms/problems. An important role that the proposed algorithm can play is
that it can be used as a benchmark for performance evaluation
of other heuristic/local optimization algorithms for the same
problem.} 

{\section*{Acknowledgment}
The authors would like to thank three anonymous reviewers for their constructive comments, which significantly improved the quality and presentation of the paper.}

\appendices

\section{An illustration of the ACR-BB algorithm}\label{app-example}
To make the proposed ACR-BB algorithm clear, an illustration of applying it to solve the following problem instance with $(N,M)=(2,3)$ is given:
$$\bh_1=\left[1.3514 +2.5260 \textbf{i}, -0.2938 -1.2571  \textbf{i}\right]^T,$$
$$\bh_2=\left[-0.2248+1.6555 \textbf{i}, -0.8479 -0.8655  \textbf{i}\right]^T,$$
$$\bh_3=\left[-0.7145-1.1201 \textbf{i}, -0.5890 +0.3075  \textbf{i}\right]^T.$$
We set $\epsilon=0.1.$

At the $0$-th iteration, we initialize $\mathcal{A}^0=\left[0,2\pi\right]^2.$ The ACR-BB algorithm solves problem ACR($\mathcal{A}^0$) and obtains its optimal solution $$\bw^0=\left[-0.3238 +0.5076\textbf{i}, -0.2669 -0.1394\textbf{i} \right],$$ $$\bc^0=\left[-1.8166 +0.2446\textbf{i}, -0.6618 -0.3010 \textbf{i}\right],$$$$L^0=0.4532.$$ Then, $\bw^0$ is scaled to obtain the feasible point $$\hat{\bw}^0=\left[-0.4454 +0.6982\textbf{i},-0.3671-0.1917 \textbf{i}\right].$$ Now, $${U^0}=\|\hat{\bw}^0\|^2=0.8573,~\bw^*=\hat{\bw}^0,$$ and $$\mathcal{P}=\left\{\{\mathcal{A}^0,\bc^0,L^0\}\right\}.$$

At the $1$-th iteration, we have $${U^1=U^0}~\text{and}~\left\{\mathcal{A}^1,\bc^1,L^1\right\}=\left\{\mathcal{A}^0,\bc^0,L^0\right\}.$$
Since $$\frac{{U^{1}}-L^{1}}{L^{1}}=\frac{0.8573-0.4532}{0.4532}=0.8917>\epsilon,$$ the ACR-BB algorithm starts executing Line \ref{line:k*}. Recall $|\bc^1|=\left[1.8330, 0.7270\right],$ where $|\cdot|$ denotes the component-wise absolute value operator. 
We have $$k^*=2,~z_2^1=\pi,$$ $$\mathcal{A}_l^1=\left[0,2\pi\right]\times \left[0,\pi\right]~\text{and}~\mathcal{A}_r^1=\left[0,2\pi\right]\times \left[\pi, 2\pi\right].$$ The ACR-BB algorithm solves problems ACR($\mathcal{A}_l^1$) and ACR($\mathcal{A}_r^1$) and obtains their optimal solutions
$$\bw_l^1=\left[-0.2558 +0.4725 \textbf{i},-0.3486 -0.2688 \textbf{i}\right],$$$$\bc_l^1=\left[-1.7747 +0.5097 \textbf{i},-0.6618 +0.0000 \textbf{i}\right],$$$$L_l^1=0.4825,$$
and
$$\bw_r^1=\left[-0.3238 +0.5076 \textbf{i},-0.2669 -0.1394 \textbf{i}\right],$$$$\bc_r^1=\left[-1.8166 +0.2446 \textbf{i},-0.6618 -0.3010 \textbf{i}\right],$$$$L_r^1=0.4532.$$
Then, $\bw_l^1$ and $\bw_r^1$ are scaled to obtain two feasible points $$\hat{\bw}_l^1=\left[-0.3864 +0.7139 \textbf{i},-0.5267 -0.4062 \textbf{i}\right],$$ $$\hat{\bw}_r^1=\left[-0.4454 +0.6982 \textbf{i},-0.3671 -0.1917 \textbf{i}\right],$$ respectively. Now, since $\|\hat{\bw}_l^1\|^2=1.1015>{U^1}$ and {$\|\hat{\bw}_r^1\|^2={0.8573}={U^1},$ ${U^1}$ and $\bw^*$} are not updated; since $L_l^1<{U^1}$ and $L_r^1<{U^1},$ we have
$$\mathcal{P}=\left\{\{\mathcal{A}_l^1,\bc_l^1,L_l^1\}, \{\mathcal{A}_r^1,\bc_r^1,L_r^1\}\right\}.$$

At the $2$-th iteration, we have ${U^2=U^1};$~since $L_r^1<L_l^1,$ we have $$\left\{\mathcal{A}_l^2,\bc_l^2,L_l^2\right\}=\left\{\mathcal{A}_r^1,\bc_r^1,L_r^1\right\}.$$
Since $$\frac{{U^{2}}-L^{2}}{L^{2}}=\frac{0.8573-0.4532}{0.4532}=0.8917>\epsilon,$$ the ACR-BB algorithm starts executing Line \ref{line:k*}. Recall $|\bc^2|=\left[1.8330, 0.7270\right].$ We have $$k^*=2,~z_2^1=\frac{3}{2}\pi,$$ $$\mathcal{A}_l^2=\left[0,2\pi\right]\times \left[\pi, \frac{3}{2}\pi\right]~\text{and}~\mathcal{A}_r^2=\left[0,2\pi\right]\times \left[\frac{3}{2}\pi, 2\pi\right].$$ The ACR-BB algorithm solves problems ACR($\mathcal{A}_l^2$) and ACR($\mathcal{A}_r^2$) and obtains their optimal solutions
$$\bw_l^2=\left[-0.3258 +0.5140 \textbf{i},-0.2539 -0.1364 \textbf{i}\right],$$$$\bc_l^2=\left[-1.8355 +0.2308 \textbf{i},-0.6804 -0.3196 \textbf{i}\right],$$ $$L_l^2=0.4534,$$
and
$$\bw_r^2=\left[-0.5569 +0.4331 \textbf{i},-0.3750 +0.3380 \textbf{i}\right],$$$$\bc_r^2=\left[-1.3117 -0.4494 \textbf{i},0.0186 -0.9814 \textbf{i}\right],$$$$L_r^2=0.7526.$$
Then $\bw_l^2$ and $\bw_r^2$ are scaled to obtain two feasible points $$\hat{\bw}_l^2=\left[-0.4334 +0.6837 \textbf{i},-0.3377 -0.1815 \textbf{i}\right],$$ $$\hat{\bw}_r^2=\left[-0.5674 +0.4413 \textbf{i},-0.3820 +0.3443 \textbf{i}\right],$$ respectively. Now, since $\|\hat{\bw}_l^2\|^2=0.8023<{U^2}$ and $\|\hat{\bw}_r^2\|^2=0.7811<{U^2},$ $U^2$ is updated to $0.7811$ and $\bw^*$ is updated to $\hat{\bw}_r^2$; since $L_l^2<{U^2}$ and $L_r^2<{U^2},$ we have
$$\mathcal{P}=\left\{\{\mathcal{A}_l^1,\bc_l^1,L_l^1\}, \{\mathcal{A}_l^2,\bc_l^2,L_l^2\}, \{\mathcal{A}_r^2,\bc_r^2,L_r^2\}\right\}.$$

At the $3$-th iteration, we have ${U^3=U^2};$ since $L_l^2<L_l^1<L_r^2,$ we have $$\left\{\mathcal{A}^3,\bc^3,L^3\right\}=\left\{\mathcal{A}_l^2,\bc_l^2,L_l^2\right\}.$$
Since $$\frac{{U^{3}}-L^{3}}{L^{3}}=\frac{0.7811-0.4534}{0.4534}=0.7228>\epsilon,$$ the ACR-BB algorithm starts executing Line \ref{line:k*}. Because $|\bc^3|=\left[1.8500, 0.7517\right],$ we have $$k^*=2,~z_2^1=\frac{5}{4}\pi,$$ $$\mathcal{A}_l^3=\left[0,2\pi\right]\times \left[\pi, \frac{5}{4}\pi\right]~\text{and}~\mathcal{A}_r^2=\left[0,2\pi\right]\times \left[\frac{5}{4}\pi, \frac{3}{2}\pi\right].$$ The ACR-BB algorithm solves problems ACR($\mathcal{A}_l^3$) and ACR($\mathcal{A}_r^3$) and obtains their optimal solutions
$$\bw_l^3=\left[-0.3196 +0.5576 \textbf{i},-0.1681 -0.1563 \textbf{i}\right],$$$$\bc_l^3=\left[-1.9876 +0.2034 \textbf{i},-0.8441 -0.3765 \textbf{i}\right],$$$$L_l^3=0.4658,$$
and
$$\bw_r^3=\left[-0.4103 +0.5652 \textbf{i},-0.1372 +0.0230 \textbf{i}\right],$$
$$\bc_r^3=\left[-1.9129 -0.1069 \textbf{i},-0.7071 -0.7071 \textbf{i}\right],$$$$L_r^3=0.5072.$$
Then $\bw_l^3$ and $\bw_r^3$ are scaled to obtain two feasible points $$\hat{\bw}_l^3=\left[-0.3458 +0.6033 \textbf{i},-0.1818 -0.1692 \textbf{i}\right],$$ $$\hat{\bw}_r^3=\left[-0.4103 +0.5652 \textbf{i},-0.1372 +0.0230 \textbf{i}\right],$$ respectively. Now, since $\|\hat{\bw}_l^3\|^2=0.5453<{U^3}$ and $\|\hat{\bw}_r^3\|^2=0.5072<{U^3},$ ${U^3}$ is updated to $0.5072$ and $\bw^*$ is updated to $\hat{\bw}_r^3;$ since $L_l^3<{U^3}$ and ${L_r^3=U^3},$ we have
$$\mathcal{P}=\left\{\{\mathcal{A}_l^1,\bc_l^1,L_l^1\}, \{\mathcal{A}_r^2,\bc_r^2,L_r^2\}, \{\mathcal{A}_l^3,\bc_l^3,L_l^3\}, \{\mathcal{A}_r^3,\bc_r^3,L_r^3\}\right\}.$$

At the $4$-th iteration, {we have $U^4=U^3;$} since $L_l^3<L_l^1<L_r^3<L_r^2,$ we have $$\left\{\mathcal{A}^4,\bc^4,L^4\right\}=\{\mathcal{A}_l^3,\bc_l^3,L_l^3\}.$$
Since $$\frac{{U^{4}}-L^{4}}{L^{4}}=\frac{0.5072-0.4658}{0.4658}=0.0889<\epsilon,$$ the ACR-BB algorithm terminates and return
$$U^*=U^4~\text{and}~\bw^*=\hat{\bw}_r^3.$$

\ifCLASSOPTIONcaptionsoff
  \newpage
\fi

\bibliographystyle{IEEEtran}
\bibliography{Beamformming-r}

\end{document}